\documentclass[a4paper]{amsart}

\title{Mathematical Properties of the Zadoff--Chu~Sequences}
\author{David~Gregoratti, Xavier~Arteaga, and~Joaquim~Broquetas}
\address{Software Radio Systems (SRS), Barcelona, Spain}
\email{\{first.last\}@srs.io}

\usepackage[T1]{fontenc}
\usepackage{amsmath}
\usepackage{amsfonts}
\usepackage{amsthm}
  \allowdisplaybreaks[2]
\usepackage{cite}
\usepackage[l2tabu,orthodox]{nag}
\usepackage{graphicx}
\usepackage{tikzexternal}
\usepackage{xcolor}
\usepackage[hidelinks]{hyperref}

\newcommand{\fig}{Figure}

\newcommand{\dei}{\,\mathrm{d}}
\newcommand{\me}{\mathrm{e}}
\newcommand{\mj}{\mathrm{j}}
\newcommand{\abs}[1]{\lvert #1 \rvert}

\newcommand{\Bigabs}[1]{\Bigl\lvert #1 \Bigr\rvert}
\newcommand{\biggabs}[1]{\biggl\lvert #1 \biggr\rvert}
\newcommand{\zz}{\mathbb{Z}}

\newcommand{\nzc}{N_{\text{\upshape ZC}}}

\newtheorem{claim}{Claim}
\newtheorem{proposition}{Proposition}

\begin{document}

\begin{abstract}
This paper is a compilation of well-known results about Zadoff--Chu sequences,
including all proofs with a consistent mathematical notation, for easy reference.
Moreover, for a Zadoff--Chu sequence $x_u[n]$ of prime length $\nzc$ and root index $u$, a
formula is derived that allows computing the first term (frequency zero) of its
discrete Fourier transform, $X_u[0]$, with constant complexity independent of
the sequence length, as opposed to accumulating all its $\nzc$ terms.
The formula stems from a famous result in analytic number theory and
is an interesting complement to the fact that the discrete Fourier transform of
a Zadoff--Chu sequence is itself a Zadoff--Chu sequence whose terms are scaled
by $X_u[0]$. Finally, the paper concludes with a brief analysis of
time-continuous signals derived from Zadoff--Chu sequences, especially those
obtained by OFDM-modulating a Zadoff--Chu sequence.
\end{abstract}

\maketitle

%%%%%%%%%%%%%%%%%%%%%%%%%%%%%%%%%%%%%%%%%%%%%%%%%%%%%%%%%%%%%%%%%%%%%%%%%%%%%%%%
%%%%%%%%%%%%%%%%%%%%%%%%%%%%%%%%%%%%%%%%%%%%%%%%%%%%%%%%%%%%%%%%%%%%%%%%%%%%%%%%
%%%%%%%%%%%%%%%%%%%%%%%%%%%%%%%%%%%%%%%%%%%%%%%%%%%%%%%%%%%%%%%%%%%%%%%%%%%%%%%%
\section{Introduction}
Sequences whose periodic autocorrelation nulls out everywhere but at the
zero-lag point play a fundamental role in several applications such as
synchronization, spread spectrum signaling, random multiple access, and radar.
A notable class of sequences with this ideal autocorrelation property consists
of the Zadoff--Chu (ZC) sequences, which were introduced by Chu \cite{J:Chu1972},
building on the works of Frank and Zadoff \cite{J:Frank1962}, as well as Heimiller
\cite{J:Heimiller1961}. Indeed, as outlined in the next section, ZC sequences
are endowed with a number of desirable features on top of their perfect
autocorrelation, and they are widely employed in the latest cellular
communications standards (namely, LTE and 5G NR) enabling functionalities like
synchronization, random access, channel sounding, and data spreading in control
channels. The interested reader is referred to \cite{arxiv:Andrews2022} and references therein
for a more detailed account of the uses of ZC sequences in the 3GPP standards.

Since modern cellular systems rely on OFDM (Orthogonal Frequency-Division
Multiplexing) modulation, the problem of computing the Discrete Fourier Transform
(DFT) of a ZC sequence is particularly relevant.
Moreover, the 3GPP specifies ZC sequences whose length $\nzc$ is a prime number.
As a result, the complexity of obliviously taking a DFT is significantly
higher than the FFT algorithm, namely $O(\nzc \log \nzc)$, see \cite{B:Tolimieri1989}.
Precomputing all the values for all possible sequences of length $\nzc$ is also
not an option, since it requires storing up to $\nzc(\nzc-1)$ complex
floating-point values, with $\nzc$ in the order of thousands. Fortunately, it can be
shown that the DFT of a ZC sequence of prime length is itself a ZC sequence
scaled both in amplitude and time~\cite{J:Beyme2009, J:Li2007}, thus reducing the complexity to
simply $O(\nzc)$.  More specifically, computing the DFT consists of two steps, both of
complexity $O(\nzc)$: first, compute the zero-frequency coefficient of the DFT
and, second, use this coefficient to scale the new ZC sequence.

Once a ZC sequence is OFDM-modulated, it results in a time-continuous signal that still
maintains some of the properties of the discrete sequence, although in a weaker
form. Specifically, the signal shows a low peak-to-average power ratio (as
compared to other OFDM signals) and its zero-lag autocorrelation is still
significantly higher than all other local maxima.

In this context, the purpose of this paper is three-fold. First, it aims to
compile several results regarding ZC sequences in a single, self-contained
document with uniform notation, which, we hope, will help the reader to have a
better overview of the topic.

The second objective is to show, in Section~\ref{sec:dft_zc}, that the
zero-frequency coefficient of the DFT of a ZC sequence can be computed by means
of a \emph{plug-and-play} formula of complexity $O(1)$, which simplifies even
further the computation and the memory requirements. Since our solution only
involves the the first of the two steps described above, the final complexity of
the DFT is still $O(\nzc)$, but, nonetheless, several times faster. As a side
note, it is worth remarking that real-time cellular systems (especially those
running as a software-defined radio platform) do benefit from all sorts of minor
optimizations on the long run.

The third and final objective is to provide, in Section~\ref{sec:continous}, an
initial characterization of continuous signals derived from ZC sequences, as the
ones resulting from their OFDM modulation.

%%%%%%%%%%%%%%%%%%%%%%%%%%%%%%%%%%%%%%%%%%%%%%%%%%%%%%%%%%%%%%%%%%%%%%%%%%%%%%%%
%%%%%%%%%%%%%%%%%%%%%%%%%%%%%%%%%%%%%%%%%%%%%%%%%%%%%%%%%%%%%%%%%%%%%%%%%%%%%%%%
%%%%%%%%%%%%%%%%%%%%%%%%%%%%%%%%%%%%%%%%%%%%%%%%%%%%%%%%%%%%%%%%%%%%%%%%%%%%%%%%
\section{Definition and Correlation Properties}\label{sec:definition}
This section is devoted to the formal definition of the ZC sequences as well as
to an overview of their main properties. Even though these properties are well
known, we report the proofs in Appendix~\ref{apdx:easy_proofs} for completeness.

In its most general form, the ZC sequence of period $\nzc$ and root $u$, with
relatively prime $\nzc, u \in \mathbb{N}$, $1 \le u \le \nzc - 1$, reads
\begin{equation}\label{eq:zc_full}
x_u[n] = \me^{-\mj \pi u \frac{n(n+c+2q)}{\nzc}},\qquad n\in\zz
\end{equation}
where $c\equiv\nzc\bmod 2$, and $q$ is any integer. One can easily show (see
Claim~\ref{claim:periodic} in Appendix~\ref{apdx:easy_proofs}) that $x_u[n]$
is indeed periodic of period $\nzc$. It is also evident that $x_u[n]$ is
unimodular.

As mentioned in the introduction, the most appealing property of the ZC
sequences concerns their autocorrelation, namely
\[
R_{u,u}[\tau] = \sum_{n=0}^{\nzc-1} x_u[n] x_u^*[n + \tau]
  = \nzc \delta_{\nzc | \tau}, \qquad \tau \in \zz
\]
where $^*$ denotes complex conjugate and where we introduced the Kronecker delta
symbol $\delta_{\nzc | \tau} = 1$ if $\tau \equiv 0 \pmod \nzc$ and zero
otherwise (see Claim~\ref{claim:autocorr} in Appendix~\ref{apdx:easy_proofs}).
This ideal autocorrelation property, together with unimodularity, categorizes ZC
sequences as CAZAC sequences (i.e., constant amplitude zero autocorrelation).

Although not zero, the cross-correlation between two ZC sequences with different
roots is also well behaved, namely
\[
\abs{R_{u,v}[\tau]} = \sum_{n=0}^{\nzc-1} x_u[n] x_v^*[n + \tau] = \sqrt{\nzc}
\]
independently of $\tau$ and as long as $u-v$ is relatively prime to~$\nzc$ (see
Claim~\ref{claim:xcorr} in Appendix~\ref{apdx:easy_proofs}).

Having defined the ZC sequences and recalled their main properties, we are now
ready to move on to the study of the DFT of a ZC sequence and to the main result
of the paper. Note that, hereafter, we will always consider the special case in
which $\nzc$ is an odd prime number and $q=0$ (indeed, it is straightforward to
show---see Claim~\ref{claim:no_q} in Appendix~\ref{apdx:easy_proofs}---that $q$
only introduces a shift and a constant rotation of the terms of the sequence
and, thus, is irrelevant for our purposes). In other words, we will use
\begin{equation}\label{eq:zc}
x_u[n] = \me^{-\mj \pi u \frac{n(n+1)}{\nzc}}
\end{equation}
where $\nzc$ is prime and $1 \le u \le \nzc - 1$ (the condition of $u$ and
$\nzc$ being coprime is now implicit).

%%%%%%%%%%%%%%%%%%%%%%%%%%%%%%%%%%%%%%%%%%%%%%%%%%%%%%%%%%%%%%%%%%%%%%%%%%%%%%%%
%%%%%%%%%%%%%%%%%%%%%%%%%%%%%%%%%%%%%%%%%%%%%%%%%%%%%%%%%%%%%%%%%%%%%%%%%%%%%%%%
%%%%%%%%%%%%%%%%%%%%%%%%%%%%%%%%%%%%%%%%%%%%%%%%%%%%%%%%%%%%%%%%%%%%%%%%%%%%%%%%
\section{The DFT of a ZC Sequence}\label{sec:dft_zc}

ZC sequences of odd prime length, viz.~(\ref{eq:zc}), benefit from an additional
property: Their DFT (and inverse DFT) is itself a ZC sequence, scaled both in
amplitude and time. Namely, as proven in, e.g., \cite{J:Beyme2009} (see also
Claim~\ref{claim:dft} in Appendix~\ref{apdx:easy_proofs}),
\begin{equation}\label{eq:dft_def}
X_u[k] = \sum_{n=0}^{\nzc-1} x_u[n] \me^{-\mj 2 \pi \frac{kn}{\nzc}}
  = X_u[0] x_u^*\bigl[u^{-1}k\bigr]
\end{equation}
where $X_u[\cdot]$ denotes the DFT of $x_u[\cdot]$ and $u^{-1}$ is the inverse
of $u$ modulo $\nzc$, that is $u^{-1} u \equiv 1 \pmod \nzc$, which is well
defined since $\nzc$ is prime.

As mentioned in the introduction, this property makes up for the diminished performance
of DFT/FFT algorithms applied to sequences of prime length \cite{B:Tolimieri1989}, thus
facilitating the adoption of ZC sequences in modern communications systems whose
physical layer relies on OFDM modulation (e.g., 5G New Radio, see \cite{TS38211}).
However, one still needs to compute the zero-frequency term of the DFT by
accumulating all terms of the original ZC sequence:
\begin{equation}\label{eq:dft0}
X_u[0] = \sum_{n=0}^{\nzc-1} x_u[n]
  = \sum_{n=0}^{\nzc-1} \me^{-\mj \pi u \frac{n(n+1)}{\nzc}}.
\end{equation}

The cost of computing this coefficient by taking the sum on the right-hand side of the
previous equation is $O(\nzc)$ and, in general, it must be repeated for all
$\nzc-1$ possible values of the root index $u$. Fortunately, the following
result, due to Gauss, provides a plug-in formula that allows us to compute the
coefficient $X_u[0]$  with a fixed number of operations. Surprisingly, there
seems to be no mention to this result in the wireless communications or signal
processing literature.

\begin{proposition}[Generalized quadratic Gauss sum]\label{prop:main}
For all prime $\nzc$, one has
\[
X_u[0] =
\ell_{2u} \eta_{\nzc} \sqrt{\nzc}\; \me^{\mj\frac{2\pi u}{\nzc}\bigl(\frac{\nzc+1}{2}\bigr)^3}
\]
where $\ell_n$ stands for the Legendre symbol\footnote{To avoid confusion with
classic fractions and parenthesis, we prefer not to use the standard notation
$(\frac{n}{\nzc})$ for the Legendre symbol of $n$ with respect to $\nzc$.} of
$n$ with respect to $\nzc$, namely
\[
\ell_n = \begin{cases}
0 & \text{if $\nzc$ divides $n$} \\
1 & \text{if $n \not\equiv 0$ is a perfect square modulo $\nzc$} \\
-1 & \text{if $n \not\equiv 0$ is not a perfect square modulo $\nzc$}
\end{cases}
\]
and where we introduced the coefficient
\[
\eta_{\nzc} = \begin{cases}
1 & \text{if $\nzc \equiv 1 \pmod 4$} \\
-\mj & \text{if $\nzc \equiv 3 \pmod 4$.}
\end{cases}
\]
\end{proposition}
\begin{proof}
See Appendix~\ref{apdx:generalized_quad_gauss}.
\end{proof}

One readily sees that the above formula only requires a small number of
operations that is independent of the length of the ZC sequence. The only minor
inconvenience arises from the Legendre symbols $\ell_{n}$, since computing each
one of them has complexity $O(\log(n)\log(\nzc))$ \cite{B:BachShallit1996}.
However, they can be computed offline and stored in a look-up table of just
$\nzc - 1$ bits, given their binary nature after excluding the trivial case
$n=0$. Such a table is extremely lightweight,
especially when compared with the amount of memory required to store all the
possible $X_u[0]$ as complex, floating point values.

\begin{figure}
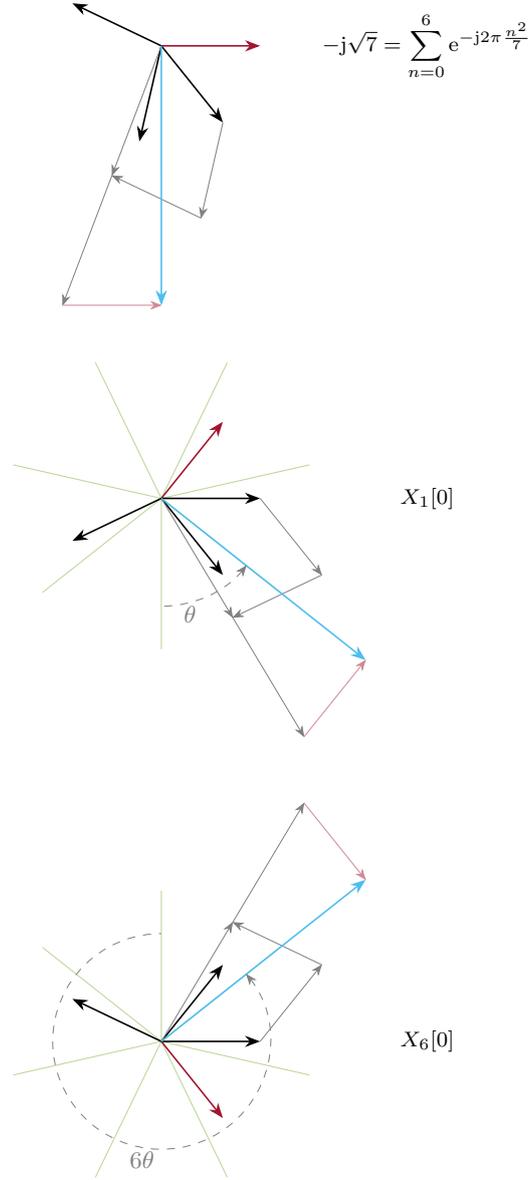

\centering
\tikz{}
\caption{Comparing $-\mj \sqrt{7} = \sum_{n=0}^{6}\me^{-\mj 2 \pi
\frac{n^2}{7}}$ (top) and $X_u[0]$, for $u=1$ (middle) and $u=6$
(bottom). The black and red vectors represent the summation
terms in (\ref{eq:dft0}).
The cyan vector is the sum of two copies of the black vectors and one copy of
the red one. The angle between each two consecutive lines in the green polar
grids is $\theta = 2 \pi (7+1)^3 / 8 / 7 \equiv 2 \pi / 7$ radians. When $u=1$,
$\ell_{2\cdot 1} = 1$ and $X_1[0]$ can be achieved as a rotation of
$-\mj \sqrt{7}$ by $\theta$. When $u=6$, $\ell_{2\cdot 6} = -1$ and $X_6[0]$
is obtained from $-\mj \sqrt{7}$ as a reflection across the horizontal axis (conjugation)
followed by a rotation of $6\theta$. Rotations are measured counterclockwise.}
\label{fig:dft}
\end{figure}

Moreover, when combined with (\ref{eq:zc}) and (\ref{eq:dft_def}),
Proposition~\ref{prop:main} implies that any term of the DFT of a ZC sequence
can be computed by a number of simple integer operations, a few real
multiplications and a single complex exponential, namely
\begin{align*}
X_u[k] &= \sqrt{\nzc} \cdot \me^{\mj \frac{\pi}{2\nzc} \chi} \\
\chi &= \biggl[4u\biggl(\frac{\nzc+1}{2}\biggr)^3 - \ell_{2u} \nzc
     + 2 k(u^{-1}k + 1)\biggr] \bmod (4\nzc)
\end{align*}
(this is the case for $\nzc \equiv 3 \pmod 4$, a similar expression can be
derived in the other case).

Another interesting interpretation of Proposition~\ref{prop:main} is that $X_u[0]$ can
be computed by taking $\ell_{2u} \eta_{\nzc} \sqrt{\nzc}$, seen as a vector in
the complex plane, and rotating it counterclockwise by $u$
steps of $2\pi(\nzc+1)^3/(8\nzc)$ radians. The procedure is sketched in
\fig~\ref{fig:dft}.

%%%%%%%%%%%%%%%%%%%%%%%%%%%%%%%%%%%%%%%%%%%%%%%%%%%%%%%%%%%%%%%%%%%%%%%%%%%%%%%%
%%%%%%%%%%%%%%%%%%%%%%%%%%%%%%%%%%%%%%%%%%%%%%%%%%%%%%%%%%%%%%%%%%%%%%%%%%%%%%%%
%%%%%%%%%%%%%%%%%%%%%%%%%%%%%%%%%%%%%%%%%%%%%%%%%%%%%%%%%%%%%%%%%%%%%%%%%%%%%%%%
\section{Time-Continuous ZC Signals}\label{sec:continous}
It is sometimes convenient to think of discrete sequences as if they were the
result of sampling a continuous signal. However, most of the properties of the ZC
sequences are intrinsic to their discrete nature and only hold in a weaker sense
in the continuous case, as we are showing next.

Given the $\nzc$-periodicity of the ZC sequences, it is only natural to look for
signals $x_u(t)$ of period $T$ such that\footnote{Discrete sequences are indexed
with brackets, continuous signals with parentheses.}
\[
x_u[n] = x_u\biggl(\frac{nT}{\nzc}\biggr).
\]

A first option would be the periodic, constant-modulus, chirp-like signal defined~by
\[
x_u^{(\textup{C})}(t) =\begin{cases}
\me^{-\mj \pi u \frac{\nzc t(\nzc t + T)}{\nzc T^2}} & t \in [0, T) \\
x_u^{(\textup{C})}(t \bmod T) &\text{otherwise}.
\end{cases}
\]
It is straightforward to show that $\lim_{t\to T^{-}} x_u^{(\textup{C})}(t) =
x_u^{(\textup{C})}(0)=1$ and that
the signal is thus continuous on $\mathbb{R}$, although not differentiable at
$t=mT$ for all \mbox{$m\in\zz$}. However, signal $x_u^{(\textup{C})}(t)$ is not a
good choice in most applications, since its instantaneous frequency reaches the
value
\[
f_{\max} = u\frac{2\nzc +1}{2T}
\]
which is higher than the sampling rate $f_0 = \nzc /T$. \fig~\ref{fig:chirp}
shows an example chirp signal, with markers placed on the samples corresponding
to the ZC sequence.

\begin{figure}
\centering
\tikz{}
\caption{Real and imaginary parts of $x_u^{(\textup{C})}(t)$ for $\nzc =5$ and
$u=4$. The markers represent the sampled values $x_u[n], n=0,1,\dots,4$.}
\label{fig:chirp}
\end{figure}

A probably more sensitive choice is to build a low-pass signal as an $\nzc$-term
Fourier sum with coefficients $X_u[n]$, the DFT of the ZC sequence $x_u[n]$.
Namely, we define
\[
x_u^{(\textup{LP})}(t) = \frac{1}{\nzc}\sum_{k=-N_0}^{N_0} X_u[k] \me^{\mj 2 \pi
\frac{kt}{T}}
\]
where we introduced $N_0=(\nzc-1)/2$ and where we adopted the usual circular-shift
convention $X_u[-k]=X_u[\nzc-k]$ for all $k=1,2,\dots,N_0$. Note that one can think
of $x_u^{(\textup{LP})}(t)$ as the result of OFDM-modulating the frequency-domain
symbols $X_u[k]$, similarly to what the 5G standard specifies for the PRACH
channel, see \cite[Sections~6.3.3 and~5.3.2]{TS38211}.

As intuition (and the example in \fig~\ref{fig:lowpass}) suggests, this signal
is much smoother than the chirp one considered before, and, by construction, has
a limited bandwidth that is compatible with the sampling rate. Nevertheless, as we
are going to see next, the properties of the ZC sequences only hold in a weaker
sense.

\begin{figure}
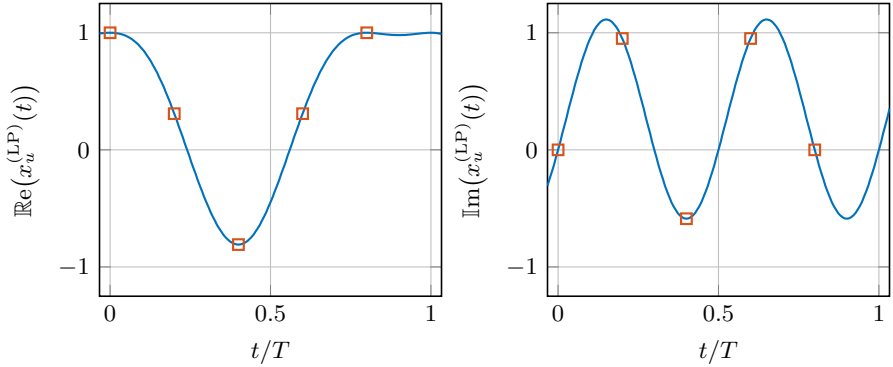

\centering
\tikz{}
\caption{Real and imaginary parts of $x_u^{(\textup{LP})}(t)$ for $\nzc=5$ and
$u=4$. The markers represent the sampled values $x_u[n], n=0,1,\dots,4$.}
\label{fig:lowpass}
\end{figure}

To start with, we readily see that the signal is not unimodular: Note how the
imaginary part of $x_u^{(\textup{LP})}(t)$ takes values larger than one on the
right-hand plot of \fig~\ref{fig:lowpass}. However, it is straightforward to
show that the signal has unitary power:
\begin{multline*}
\frac{1}{T}\int_0^T \abs{x_u^{(\textup{LP})}(t)}^2 \dei t
=\frac{1}{\nzc^2 T}\sum_{k=-N_0}^{N_0} \sum_{l=-N_0}^{N_0} X_u[k] X_u^*[l]
  \int_0^T \me^{\mj 2 \pi \frac{(k-l)t}{T}} \dei t \\
=\frac{1}{\nzc^2}\sum_{k=-N_0}^{N_0} \abs{X_u[k]}^2 = 1.
\end{multline*}

In Appendix~\ref{apdx:time_papr}, we also show that the peak amplitude of
$x_u^{(\textup{LP})}(t)$ increases at most as $O(\ln \nzc)$ with the sequence
length $\nzc$. This bound, which we believe can be made even tighter, already
represents a significant improvement with respect to the plain DFT one of
$O(\sqrt{\nzc})$, thus supporting the use of OFDM-modulated ZC sequences for
signaling applications that require a low peak-to-average power ratio
\cite{arxiv:Andrews2022}.

\begin{figure}
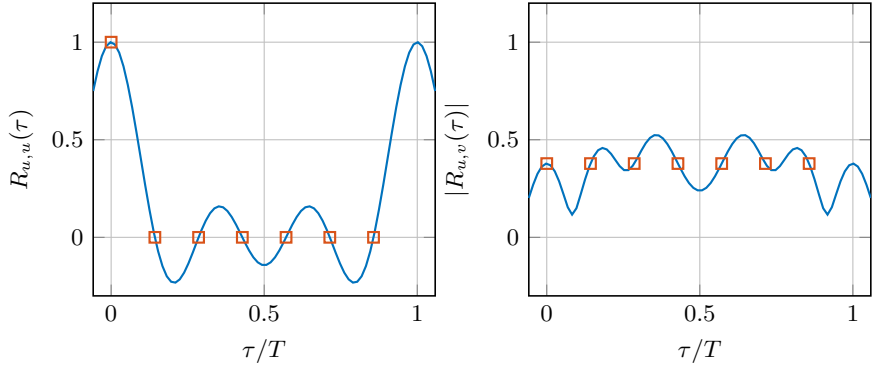

\centering
\tikz{}
\caption{Autocorrelation (left) and cross correlation (absolute value, right) of
time continuous ZC signals, for $\nzc=7$. The markers represent the discrete
case $R_{u,u}[n]$ and $\abs{R_{u,v}[n]}$, for $n=0,1,\dots,6$.}\label{fig:corr}
\end{figure}

The perfect autocorrelation property is also lost when moving to time-continuous
signals. Indeed, as we show in Appendix~\ref{apdx:time_xcorr},
\[
R_{u,u}(\tau) = \frac{1}{T}\int_0^T x_u^{(\textup{LP})}(t)
\bigl(x_u^{(\textup{LP})}(t+\tau)\bigr)^* \dei t
=D_{\nzc}\Bigl(\frac{\tau}{T}\Bigr)
\]
where we have introduced the Dirichlet kernel
\begin{equation}\label{eq:dirichlet}
D_{\nzc}(z) = \frac{1}{\nzc}\sum_{k=-N_0}^{N_0} \me^{-\mj 2 \pi k z}
= \begin{cases}
1 & \text{if } z\in \zz \\
\frac{\sin(\pi \nzc z)}{\nzc \sin(\pi z)} &\text{otherwise.}
\end{cases}
\end{equation}
In plain words, $|R_{u,u}(\tau)|$ still takes its peak value of one when
$\tau = l T$ and vanishes for $\tau = \frac{mT}{\nzc} + l T$, for all $l \in
\zz$ and $m=1,2,\dots,\nzc -1$, but it is not zero for all other values of
$\tau$. An example is reported in the left graph of \fig~\ref{fig:corr}, where
we also marked the points corresponding to the discrete case.\footnote{Note
that, for the continuous case, we normalize the correlations by the period $T$,
while we do not introduce any normalization in the discrete case. Because of
this, the continuous version is scaled by a factor $1/\nzc$ with respect to the
discrete case.}

The cross-correlation between ZC signals with different roots is given by
\[
R_{u,v}(\tau) = \frac{1}{T}\int_0^T x_u^{(\textup{LP})}(t)
\bigl(x_v^{(\textup{LP})}(t+\tau)\bigr)^* \dei t.
\]
It is relatively simple to show that its modulus takes the value $1/\sqrt{\nzc}$
when $\tau=nT/\nzc$, $n=0,1,\dots,\nzc-1$, but a full analysis is not as
straightforward. However, the right graph of \fig~\ref{fig:corr} shows that the
modulus is not constant outside those points.

%%%%%%%%%%%%%%%%%%%%%%%%%%%%%%%%%%%%%%%%%%%%%%%%%%%%%%%%%%%%%%%%%%%%%%%%%%%%%%%%
%%%%%%%%%%%%%%%%%%%%%%%%%%%%%%%%%%%%%%%%%%%%%%%%%%%%%%%%%%%%%%%%%%%%%%%%%%%%%%%%
%%%%%%%%%%%%%%%%%%%%%%%%%%%%%%%%%%%%%%%%%%%%%%%%%%%%%%%%%%%%%%%%%%%%%%%%%%%%%%%%
\section{Summary}
In this paper we have revisited some fundamental properties that make ZC
sequences very appealing for numerous applications in wireless communications.
Specifically, using a consistent mathematical notation, we have reported results
from several different authors, which show that ZC sequences are endowed with
perfect autocorrelation and very low cross-correlation
(see Section~\ref{sec:definition}).

Similarly, Section~\ref{sec:dft_zc} summarizes results related with taking the
DFT of a ZC sequence, which turns out to be a very cheap operation in spite of
the sequence length being a prime number. The section ends by providing a very
simple expression for the zero-indexed coefficient of the DFT [or, equivalently,
for the sum of all the terms of the ZC sequence, see~(\ref{eq:dft0})]: Even
though this result follows from the well-known quadratic Gauss sum, there seems
to be no mention of it in the technical literature, at least to the best of
these authors' knowledge.

Finally, Section~\ref{sec:continous} shows how to obtain a time-continuous
version of a ZC sequence. Although the resulting signal is still well behaved,
the price to pay is a mild crest factor and a loss of sharpness in the
correlation properties.

\appendix

%%%%%%%%%%%%%%%%%%%%%%%%%%%%%%%%%%%%%%%%%%%%%%%%%%%%%%%%%%%%%%%%%%%%%%%%%%%%%%%%
%%%%%%%%%%%%%%%%%%%%%%%%%%%%%%%%%%%%%%%%%%%%%%%%%%%%%%%%%%%%%%%%%%%%%%%%%%%%%%%%
%%%%%%%%%%%%%%%%%%%%%%%%%%%%%%%%%%%%%%%%%%%%%%%%%%%%%%%%%%%%%%%%%%%%%%%%%%%%%%%%
\section{Proofs of Basic Results}\label{apdx:easy_proofs}

The results of this appendix are well known in the literature (see, e.g.,
\cite{J:Chu1972, J:Frank1962, J:Heimiller1961, J:Beyme2009} and references
therein). We report them here for the sake of having a self-contained summary
document with uniform notation.

We start by showing periodicity of ZC sequences.

\begin{claim}\label{claim:periodic}
For all $u$, the ZC sequence $x_u[n]$ in (\ref{eq:zc_full}) is periodic of period $\nzc$.
\end{claim}

\begin{proof}
For all $k\in \mathbb{Z}$
\begin{multline*}
x_u[n + k\nzc] = \me^{-\mj \pi u \frac{(n+k\nzc) (n+k\nzc+c+2q)}{\nzc}} \\
= \me^{-\mj \pi u \frac{n (n+c+2q) + (2n+2q+c)k\nzc + k^2\nzc^2}{\nzc}} \\
= \me^{-\mj \pi u \frac{n (n+c+2q)}{\nzc} - \mj \pi u k[2n+2q+c + k\nzc]}
= x_u[n].
\end{multline*}
The last step is straightforward when $k$ is even, while, when $k$ is odd, it
stems from $c + k\nzc$ being even, since $c\equiv\nzc\bmod 2$.
\end{proof}

Next, we prove two claims about autocorrelation and cross-correlation of ZC
sequences.
\begin{claim}\label{claim:autocorr}
For all $\tau\in\mathbb{Z}$,
\[
R_{u,u}[\tau] = \sum_{n=0}^{\nzc-1} x_u[n] x_u^*[n + \tau] = \nzc \delta_{\nzc | \tau}
\]
where $\delta_{\nzc | \tau} = 1$ if $\tau \equiv 0 \pmod \nzc$ and zero otherwise.
\end{claim}
\begin{proof}
By the definition of $x_u[n]$, one obtains
\begin{align*}
R_{u,u}[\tau] 
&= \sum_{n=0}^{\nzc-1} \me^{-\mj \pi u \frac{n(n+c+2q)-(n+\tau)(n+\tau+c+2q)}{\nzc}} \\
&= \sum_{n=0}^{\nzc-1} \me^{\mj \pi u \frac{\tau(\tau+c+2q)+2n\tau}{\nzc}} \\
&=  \me^{\mj \pi u \frac{\tau(\tau+c+2q)}{\nzc}}
  \sum_{n=0}^{\nzc-1} \me^{\mj 2 \pi u \frac{n\tau}{\nzc}} \\
&= \nzc \delta_{\nzc | \tau}
\end{align*}
given the well-known result
\[
\sum_{n=0}^{\nzc-1} \me^{\mj 2 \pi u \frac{n\tau}{\nzc}} =
\begin{cases}
\nzc &\text{if $\nzc$ divides $\tau$} \\
0 &\text{otherwise}
\end{cases}
\]
and the fact that $\tau(\tau+c+2q)$ is always even when $\tau$ is a multiple of
$\nzc$.
\end{proof}

\begin{claim}\label{claim:xcorr}
For all $u,v=1,2,\dots,\nzc-1$, such that $u - v$ is relatively prime to $\nzc$,
and for all $\tau\in\mathbb{Z}$,
\[
\abs{R_{u,v}[\tau]} = \sum_{n=0}^{\nzc-1} x_u[n] x_v^*[n + \tau] = \sqrt{\nzc}.
\]
\end{claim}
\begin{proof}
Let
\begin{align*}
x_u[n] &= \me^{-\mj \pi \frac{un(n+c+2q_u)}{\nzc}} &
x_v[n] &= \me^{-\mj \pi \frac{vn(n+c+2q_v)}{\nzc}}.
\end{align*}
Thus, expanding $x_u[\cdot]$ and $x_v[\cdot]$ in the expression of
$R_{u,v}[\tau]$, we obtain
\begin{align*}
R_{u,v}[\tau]
&= \sum_{n=0}^{\nzc-1} \me^{-\mj \pi
  \frac{un(n+c+2q_u)-v(n+\tau)(n+\tau+c+2q_v)}{\nzc}} \\
&= \sum_{n=0}^{\nzc-1} \me^{-\mj \pi \frac{(u-v)n(n+c) + 2 u n q_u -2 v n
  (\tau + q_v) - v \tau (\tau+c+2q_v)}{\nzc}} \\
&= \me^{\mj \pi v \frac{\tau (\tau+c+2q_v)}{\nzc}}
  \sum_{n=0}^{\nzc-1} \me^{-\mj \pi \frac{(u-v)n(n+c)+2n(u q_u - v\tau - v q_v)}{\nzc}} \\
&= \me^{\mj \pi v \frac{\tau (\tau+c+2q_v)}{\nzc}}
  \sum_{n=0}^{\nzc-1} \me^{-\mj \pi (u-v) \frac{n[n+c+2\zeta v(u q_u - v\tau - v q_v)]}{\nzc}} \\
&= \me^{\mj \pi v \frac{\tau (\tau+c+2q_v)}{\nzc}}
  \sum_{n=0}^{\nzc-1} x_{u-v}[n]
\end{align*}
where we have introduced $\zeta = (u-v)^{-1}$ modulo $\nzc$ (recall that, by
assumption, $u-v$ and $\nzc$ are coprime) and where $x_{u-v}[n]$ is a new ZC
sequence with root $u-v$, where $\zeta v(u q_u - v\tau - v q_v)$ plays the role
of the phase parameter $q$. Incidentally, note that the last step also shows
that the cross-correlation between $x_u[n]$ and $x_v[n]$ is nothing else but a
scaled and conjugated version of $x_v[n]$, that is
\begin{equation}\label{eq:xcorr_nice}
R_{u,v}[\tau] = x_v^*[\tau] \sum_{n=0}^{\nzc-1} x_{u-v}[n].
\end{equation}

To finish the proof, one can easily compute the (squared) modulus of the previous
expression. Indeed, after letting $q=\zeta v(u q_u - v\tau - v q_v)$,
\begin{align*}
\abs{R_{u,v}[\tau]}^2 
  &= \sum_{n=0}^{\nzc-1} \me^{-\mj \pi (u-v) \frac{n(n+c+2q)}{\nzc}}
     \sum_{k=0}^{\nzc-1} \me^{\mj \pi (u-v) \frac{k(k+c+2q)}{\nzc}} \\
  &= \sum_{n=0}^{\nzc-1} \sum_{k=0}^{\nzc-1} \me^{-\mj \pi (u-v) \frac{n(n+c+2q)}{\nzc}}
     \me^{\mj \pi (u-v) \frac{k(k+c+2q)}{\nzc}} \\
\intertext{and, since periodicity allows the change of variable $k \mapsto k+n$,}
  &= \sum_{n=0}^{\nzc-1} \sum_{k=0}^{\nzc-1} \me^{-\mj \pi (u-v) \frac{n(n+c+2q)}{\nzc}}
     \me^{\mj \pi (u-v) \frac{(k+n)(k+n+c+2q)}{\nzc}} \\
  &= \sum_{n=0}^{\nzc-1} \sum_{k=0}^{\nzc-1} \me^{\mj \pi (u-v) \frac{k(k+2n+c+2q)}{\nzc}} \\
  &= \sum_{k=0}^{\nzc-1} \me^{\mj \pi (u-v) \frac{k(k+c+2q)}{\nzc}}
      \sum_{n=0}^{\nzc-1} \me^{\mj 2 \pi (u-v) \frac{kn}{\nzc}} \\
  &= \nzc \sum_{k=0}^{\nzc-1} \me^{\mj \pi (u-v) \frac{k(k+c+2q)}{\nzc}} \delta_{k} \\
  &= \nzc.
\end{align*}

As a last remark, note that the combination of the last step and
(\ref{eq:xcorr_nice}) also tells us that
\[
\biggabs{\sum_{n=0}^{\nzc-1} x_{u}[n]} = \sqrt{\nzc}
\]
for all $\nzc$ and $u$ relatively prime. Proposition~\ref{prop:main} in
Section~\ref{sec:dft_zc} allows us to
be more precise when $\nzc$ is a prime number.
\end{proof}

We now prove a simple result that justifies focusing on the simplified ZC
sequences in (\ref{eq:zc}), instead of the general form in (\ref{eq:zc_full}).
\begin{claim}\label{claim:no_q}
Let
\begin{align*}
x_u[n] &= \me^{-\mj \pi u \frac{n(n+c)}{\nzc}} &&\text{and} &
x_u^{(q)}[n] &= \me^{-\mj \pi u \frac{n(n+c+2q)}{\nzc}}.
\end{align*}
Then
\[
x_u^{(q)}[n] = x_u[n+q] \me^{\mj \pi u \frac{q(q+c)}{\nzc}}.
\]
\end{claim}
\begin{proof}
With straightforward manipulations:
\begin{align*}
x_u^{(q)}[n] &= \me^{-\mj \pi u \frac{n(n+c+2q)}{\nzc}} \\
  &= \me^{-\mj \pi u \frac{(n+q)(n+c+2q)-q(n+c+2q)}{\nzc}} \\
  &= \me^{-\mj \pi u \frac{(n+q)(n+c+q) +q(n+q)-q(n+c+2q)}{\nzc}} \\
  &= \me^{-\mj \pi u \frac{(n+q)(n+c+q)}{\nzc}}
      \me^{\mj \pi u \frac{q(q+c)}{\nzc}} \\
  &= x_u[n+q] \me^{\mj \pi u \frac{q(q+c)}{\nzc}}.
\end{align*}
\end{proof}

Finally, we show how to compute the DFT of a ZC sequence.

\begin{claim}[see \cite{J:Beyme2009}]\label{claim:dft}
The DFT of a ZC sequence of odd prime length $\nzc$, viz.~(\ref{eq:zc}), is
itself a ZC sequence, scaled both in amplitude and time. Namely,
\[
X_u[k] = \sum_{n=0}^{\nzc-1} x_u[n] \me^{-\mj 2 \pi \frac{kn}{\nzc}}
  = X_u[0] x_u^*\bigl[u^{-1}k\bigr]
\]
where $u^{-1}$ is the inverse of $u$ modulo $\nzc$, that is $u^{-1} u \equiv 1
\pmod \nzc$, which is well defined since $\nzc$ is prime.

Moreover, the time-scaled ZC sequence is equivalent to a modulation of the ZC
sequence of root index $u^{-1}$:
\[
x_u^*[u^{-1}k] = x_{u^{-1}}^*[k] \me^{-\mj \pi k\frac{(u^{-1}-1)(\nzc+1)}{\nzc}}.
\]
\end{claim}
\begin{proof}
As usual, we expand $x_u[n]$ to obtain
\begin{align*}
X_u[k] &= \sum_{n=0}^{\nzc-1} x_u[n] \me^{-\mj 2 \pi \frac{kn}{\nzc}} \\
&= \sum_{n=0}^{\nzc-1} \me^{-\mj \pi u \frac{n(n+1) + 2u^{-1} kn }{\nzc}} \\
\intertext{and, applying the change of variable $n \mapsto n-u^{-1}k$ and
exploiting periodicity,}
&= \sum_{n=0}^{\nzc-1} \me^{-\mj \pi u \frac{(n-u^{-1}k)(n-u^{-1}k+1)
  + 2u^{-1} k(n-u^{-1}k) }{\nzc}} \\
&= \sum_{n=0}^{\nzc-1} \me^{-\mj \pi u \frac{n^2 + n -u^{-2}k^2 - u^{-1} k }{\nzc}} \\
&= \me^{\mj \pi u \frac{u^{-1}k(u^{-1} k + 1) }{\nzc}}
    \sum_{n=0}^{\nzc-1} \me^{-\mj \pi u \frac{n(n+1)}{\nzc}} \\
&= x_u^*\bigl[u^{-1}k\bigr] X_u[0].
\end{align*}

The proof of the second statement is also straightforward. Let $\alpha =
\frac{\nzc+1}{2}$ be the inverse of 2 modulo $\nzc$. Then, since $u^{-1}k(u^{-1}
k + 1)$ is even,
\begin{multline*}
x_u^*\bigl[u^{-1}k\bigr] = \me^{\mj \pi u \frac{u^{-1}k(u^{-1} k + 1) }{\nzc}}
= \me^{\mj 2 \pi u \alpha \frac{u^{-1}k(u^{-1} k + 1) }{\nzc}} \\
= \me^{\mj 2\pi u^{-1} \alpha \frac{k(k + u) }{\nzc}}
= \me^{\mj 2\pi u^{-1} \alpha \frac{k(k + 1) }{\nzc}} \me^{\mj 2 \pi u^{-1}
  \alpha \frac{k(u - 1) }{\nzc}} \\
= \me^{\mj \pi u^{-1} \frac{k(k + 1) }{\nzc}} \me^{-\mj 2 \pi \alpha
  \frac{k(u^{-1} - 1) }{\nzc}}
= x_{u^{-1}}^*[k] \me^{-\mj \pi (\nzc + 1)\frac{k(u^{-1}-1)}{\nzc}}.
\end{multline*}
\end{proof}

%%%%%%%%%%%%%%%%%%%%%%%%%%%%%%%%%%%%%%%%%%%%%%%%%%%%%%%%%%%%%%%%%%%%%%%%%%%%%%%%
%%%%%%%%%%%%%%%%%%%%%%%%%%%%%%%%%%%%%%%%%%%%%%%%%%%%%%%%%%%%%%%%%%%%%%%%%%%%%%%%
%%%%%%%%%%%%%%%%%%%%%%%%%%%%%%%%%%%%%%%%%%%%%%%%%%%%%%%%%%%%%%%%%%%%%%%%%%%%%%%%
\section{Proof of Proposition~\ref{prop:main}}\label{apdx:generalized_quad_gauss}

Proposition~\ref{prop:main} follows straightforwardly from well known
number-theoretic results. We summarize here the main steps.

We start by denoting $\alpha = (\nzc + 1)/2$. Observing that $\alpha \equiv 2^{-1}
\pmod{\nzc}$ and that $n(n+1)$ is even, we can write
\[
X_u[0] = \sum_{n=0}^{\nzc-1} \me^{-\mj \pi u \frac{n(n+1)}{\nzc}}
= \sum_{n=0}^{\nzc-1} \me^{-\mj 2 \pi u \frac{\alpha n(n+1)}{\nzc}}
\]
and, completing the square at the exponent numerator,
\begin{align*}
X_u[0] &= \sum_{n=0}^{\nzc-1} \me^{-\mj 2 \pi u
  \frac{\alpha (n+\alpha)^2 - \alpha^3}{\nzc}} \\
&= \me^{\mj 2 \pi u \frac{\alpha^3}{\nzc}}
  \sum_{n=0}^{\nzc-1} \me^{-\mj 2 \pi u \frac{\alpha (n+\alpha)^2}{\nzc}} \\
&= \me^{\mj 2 \pi u \frac{\alpha^3}{\nzc}}
  \sum_{n=0}^{\nzc-1} \me^{-\mj 2 \pi u \frac{\alpha n^2}{\nzc}}
\end{align*}
where the last step is due to the fact that, modulo $\nzc$, $n + \alpha$ runs
over all values in $\{0, 1, \dots, \nzc-1\}$, same as $n$.

To conclude the proof, we only need the two following well-known results
regarding quadratic Gauss sums.

\begin{proposition}[{\cite[Chapter~6, Theorem~1]{B:Ireland1982}}]\label{prop:gauss_simple}
Let $g_{\nzc} = \sum_{n=0}^{\nzc-1} \me^{\mj 2 \pi \frac{n^2}{\nzc}}$ denote the
quadratic Gauss sum. Then
\[
g_{\nzc} = \begin{cases}
\sqrt{\nzc} & \text{if } \nzc \equiv 1 \pmod{4} \\
\mj\sqrt{\nzc} & \text{if } \nzc \equiv 3 \pmod{4}
\end{cases}
\]
\end{proposition}

\begin{proposition}[{\cite[Chapter~6, Proposition~6.3.1]{B:Ireland1982}}]\label{prop:gauss_b}
For all $b \in \zz$ that is not a multiple of $\nzc$,
\[
\sum_{n=0}^{\nzc-1} \me^{\mj 2 \pi b \frac{n^2}{\nzc}} = \ell_{b} g_{\nzc}
\]
where, we recall, $\ell_b$ is the Legendre symbol of $b$ with respect to $\nzc$.
\end{proposition}

Indeed, by applying Proposition~\ref{prop:gauss_b} with $b=-u\alpha$, we have
\[
X_u[0] = \me^{\mj 2 \pi u \frac{\alpha^3}{\nzc}}
  \sum_{n=0}^{\nzc-1} \me^{-\mj 2 \pi u \frac{\alpha n^2}{\nzc}}
= \me^{\mj 2 \pi u \frac{\alpha^3}{\nzc}} \ell_{-u\alpha} g_{\nzc}
= \me^{\mj 2 \pi u \frac{\alpha^3}{\nzc}} \ell_{-1}\ell_{2u} g_{\nzc}
\]
where we have used the fact that the Legendre symbol is multiplicative,
therefore $\ell_{-u\alpha}=\ell_{-1}\ell_{u\alpha}$, together with $\ell_{\alpha} =
\ell_{2}$, since $\alpha$ is the inverse of $2$. The desired result follows from
Proposition~\ref{prop:gauss_simple}, after recalling that
\[
\ell_{-1} = \begin{cases}
1 & \text{if } \nzc \equiv 1 \pmod{4} \\
-1 & \text{if } \nzc \equiv 3 \pmod{4}.
\end{cases}
\]

%%%%%%%%%%%%%%%%%%%%%%%%%%%%%%%%%%%%%%%%%%%%%%%%%%%%%%%%%%%%%%%%%%%%%%%%%%%%%%%%
%%%%%%%%%%%%%%%%%%%%%%%%%%%%%%%%%%%%%%%%%%%%%%%%%%%%%%%%%%%%%%%%%%%%%%%%%%%%%%%%
%%%%%%%%%%%%%%%%%%%%%%%%%%%%%%%%%%%%%%%%%%%%%%%%%%%%%%%%%%%%%%%%%%%%%%%%%%%%%%%%
\section{Proofs Regarding Time-Continuous Signals}
This appendix is devoted to proving the results about time-continuous ZC
signals. Since we will exclusively focus on the low-pass option, in what follows
we adopt a simplified notation and write $x_u(t)$ instead of
$x_u^{(\mathrm{LP})}(t)$.

%%%%%%%%%%%%%%%%%%%%%%%%%%%%%%%%%%%%%%%%%%%%%%%%%%%%%%%%%%%%%%%%%%%%%%%%%%%%%%%%
%%%%%%%%%%%%%%%%%%%%%%%%%%%%%%%%%%%%%%%%%%%%%%%%%%%%%%%%%%%%%%%%%%%%%%%%%%%%%%%%
\subsection{Peak Amplitude of $x_u(t)$}\label{apdx:time_papr}
To study the amplitude of $x_u(t)$, let us first express it in terms of the
discrete sequence $x_u[n]$:
\begin{align*}
x_u(t) &= \frac{1}{\nzc}\sum_{k=-N_0}^{N_0} X_u[k] \me^{\mj 2\pi\frac{kt}{T}} \\
&= \frac{1}{\nzc}\sum_{k=-N_0}^{N_0} \sum_{n=0}^{\nzc-1} x_u[n]
  \me^{-\mj 2\pi \frac{kn}{\nzc}} \me^{\mj 2\pi\frac{kt}{T}} \\
&= \frac{1}{\nzc} \sum_{n=0}^{\nzc-1} x_u[n] \sum_{k=-N_0}^{N_0}
  \me^{\mj 2\pi k\bigl(\frac{t}{T} - \frac{n}{\nzc}\bigr)} \\
&= \sum_{n=0}^{\nzc-1} x_u[n] D_{\nzc}\Bigl(\frac{t}{T}-\frac{n}{\nzc}\Bigr)
\end{align*}
where $D_{\nzc}(\cdot)$ denotes the Dirichlet kernel as defined in
(\ref{eq:dirichlet}). Then, by the triangle inequality,
\[
\abs{x_u(t)} \le \sum_{n=0}^{\nzc-1}
\Bigabs{D_{\nzc}\Bigl(\frac{t}{T}-\frac{n}{\nzc}\Bigr)}.
\]

Building upon symmetry arguments, one can show that the right-hand side of
the above inequality has period $T/\nzc$ and has maximum points at
$T\frac{2m+1}{2\nzc}$, for $m\in\zz$.
Therefore, we can choose $m=-1$, set $t=-T/(2\nzc)$ on the right-hand side and
exploit the fact that $D_{\nzc}(\cdot)$ is an even function to write
\begin{align*}
\abs{x_u(t)} &\le \sum_{n=0}^{\nzc-1}
  \Bigabs{D_{\nzc}\Bigl(\frac{2n+1}{2\nzc}\Bigr)} \\
&= \sum_{n=0}^{\nzc-1} \frac{\abs{\sin(\pi\frac{2n+1}{2})}}
  {\nzc\sin(\pi\frac{2n+1}{2\nzc})} \\
&= \sum_{n=0}^{\nzc-1} \frac{1}{\nzc\sin(\pi\frac{2n+1}{2\nzc})}.
\end{align*}
Note that there is no need for taking the absolute value at the denominator, since the
argument of the sine function is always inside the interval $[0, \pi)$.

Now, it remains to apply the well-known inequality $\sin z \ge z(1-z/\pi)$,
which holds for $0\le z \le \pi$, and write
\[
\abs{x_u(t)} \le \frac{4\nzc}{\pi}\sum_{n=0}^{\nzc-1}\frac{1}{(1+2n)(2\nzc - 1- 2n)}.
\]
Partial fraction decomposition and a couple of standard summation tricks allow
further simplification:
\begin{align*}
\abs{x_u(t)} &\le \frac{2}{\pi}\sum_{n=0}^{\nzc-1}\frac{1}{1+2n}
  + \frac{2}{\pi}\sum_{n=0}^{\nzc-1} \frac{1}{2\nzc - 1- 2n} \\
&= \frac{4}{\pi}\sum_{n=0}^{\nzc-1}\frac{1}{1+2n}
  \qquad\qquad\text{[$n\mapsto \nzc -1-n$ in the second sum]} \\
&= \frac{4}{\pi}\Biggl(\sum_{n=1}^{2\nzc-1}\frac{1}{n} -
  \sum_{n=1}^{\nzc-1}\frac{1}{2n}\Biggr).
\end{align*}

\begin{figure}
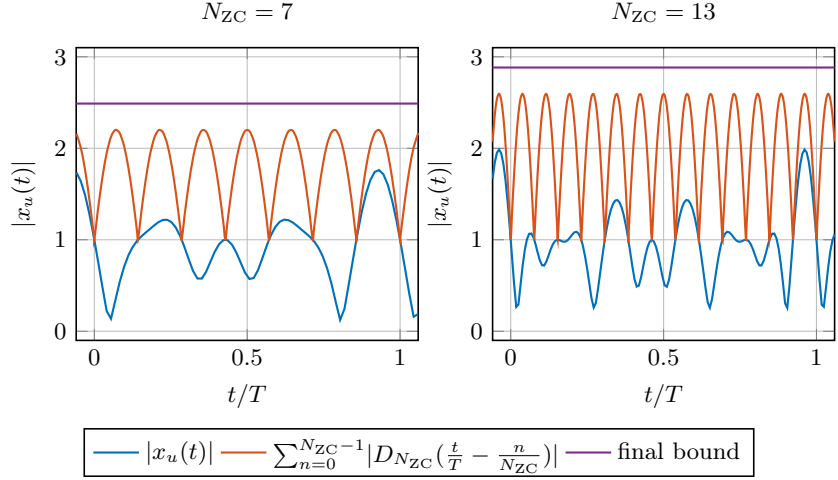

\centering
\tikz{} \\
\tikz{}
\caption{Behavior of $\abs{x_u(t)}$, of the intermediate bound
$\sum_{n=0}^{\nzc-1}\abs{D_{\nzc}(\frac{t}{T} - \frac{n}{\nzc})}$ and of the
final bound $\frac{2}{\pi}\ln(\nzc) + \frac{4}{\pi}\ln(2) +
\frac{2}{\pi}\gamma$, for $\nzc=7$ (left) and $\nzc=13$ (right). The root
sequence index is $u=4$ and $u=8$, respectively, which are the values that
maximize $\max_t\abs{x_u(t)}$.}\label{fig:mod_time}
\end{figure}

Finally, we can apply well-known results about harmonic numbers and show that,
for large $\nzc$,
\begin{align*}
\abs{x_u(t)} &\le \frac{4}{\pi}\ln(2\nzc) - \frac{2}{\pi}\ln(\nzc) +
\frac{2}{\pi}\gamma \\
&= \frac{2}{\pi}\ln(\nzc) + \frac{4}{\pi}\ln(2) + \frac{2}{\pi}\gamma
\end{align*}
where $\gamma$ is the Euler--Mascheroni constant. The final bound, as well as
the intermediate one $\sum_{n=0}^{\nzc-1}\abs{D_{\nzc}(\frac{t}{T} -
\frac{n}{\nzc})}$, is compared to $\abs{x_u(t)}$ if \fig~\ref{fig:mod_time}, for
different values of the sequence length $\nzc$.

%%%%%%%%%%%%%%%%%%%%%%%%%%%%%%%%%%%%%%%%%%%%%%%%%%%%%%%%%%%%%%%%%%%%%%%%%%%%%%%%
%%%%%%%%%%%%%%%%%%%%%%%%%%%%%%%%%%%%%%%%%%%%%%%%%%%%%%%%%%%%%%%%%%%%%%%%%%%%%%%%
\subsection{Autocorrelation of $x_u(t)$}\label{apdx:time_xcorr}
The autocorrelation of $x_u(t)$ is fairly easy to compute and only requires few
steps:
\begin{align*}
R_{u,u}(\tau) &= \frac{1}{T}\int_0^T x_u(t) x_u^*(t+\tau) \dei t \\
&\stackrel{\textup{(a)}}{=} \frac{1}{\nzc^2} \sum_{k=-N_0}^{N_0} \sum_{l=-N_0}^{N_0} X_u[k] X_u^*[l]
  \me^{-\mj 2 \pi \frac{l \tau}{T}} \frac{1}{T}
  \int_0^T \me^{\mj 2\pi \frac{(k-l)t}{T}}\dei t \\
&= \frac{1}{\nzc^2} \sum_{k=-N_0}^{N_0} \sum_{l=-N_0}^{N_0} X_u[k] X_u^*[l]
  \me^{-\mj 2 \pi \frac{l \tau}{T}} \delta_{k,l} \\
&= \frac{1}{\nzc^2} \sum_{k=-N_0}^{N_0} \abs{X_u[k]}^2
  \me^{-\mj 2 \pi \frac{k \tau}{T}} \\
&\stackrel{\textup{(b)}}{=} \frac{1}{\nzc} \sum_{k=-N_0}^{N_0}
  \me^{-\mj 2 \pi \frac{k \tau}{T}} \\
&= \begin{cases}
1 &\text{if } \tau = rT, r\in \zz \\
\displaystyle \frac{\sin\Bigl(\pi \nzc \frac{\tau}{T}\Bigr)}{\nzc\sin\Bigl(\pi\frac{\tau}{T}\Bigr)}
  &\text{otherwise}
\end{cases}
\end{align*}
where, in (a), we just expanded $x_u(t)$, and, in (b), we used the fact that
$\abs{X_u[n]} = \sqrt{\nzc}$, which follows directly from (\ref{eq:dft_def}) and
Proposition~\ref{prop:main}.

\bibliographystyle{IEEEtran}
\bibliography{IEEEabrv,biblio}
\end{document}